\documentclass{sig-alternate}
\makeatletter
\newif\if@borderstar
\def\bordermatrix{\@ifnextchar*{%
  \@borderstartrue\@bordermatrix@i}{\@borderstarfalse\@bordermatrix@i*}%
}
\def\@bordermatrix@i*{\@ifnextchar[{%
  \@bordermatrix@ii}{\@bordermatrix@ii[()]}
}
\def\@bordermatrix@ii[#1]#2{%
  \begingroup
    \m@th\@tempdima8.75\p@\setbox\z@\vbox{%
      \def\cr{\crcr\noalign{\kern 2\p@\global\let\cr\endline }}%
      \ialign {$##$\hfil\kern 2\p@\kern\@tempdima & \thinspace %
      \hfil $##$\hfil && \quad\hfil $##$\hfil\crcr\omit\strut %
      \hfil\crcr\noalign{\kern -\baselineskip}#2\crcr\omit %
      \strut\cr}}%
    \setbox\tw@\vbox{\unvcopy\z@\global\setbox\@ne\lastbox}%
    \setbox\tw@\hbox{\unhbox\@ne\unskip\global\setbox\@ne\lastbox}%
    \setbox\tw@\hbox{%
      $\kern\wd\@ne\kern -\@tempdima\left\@firstoftwo#1%
        \if@borderstar\kern2pt\else\kern -\wd\@ne\fi%
      \global\setbox\@ne\vbox{\box\@ne\if@borderstar\else\kern 2\p@\fi}%
      \vcenter{\if@borderstar\else\kern -\ht\@ne\fi%
        \unvbox\z@\kern-\if@borderstar2\fi\baselineskip}%
        \if@borderstar\kern-2\@tempdima\kern2\p@\else\,\fi\right\@secondoftwo#1 $%
    }\null \;\vbox{\kern\ht\@ne\box\tw@}%
  \endgroup
}
\makeatother

\usepackage[naturalnames]{hyperref} 
\usepackage{enumerate}
\usepackage{mathrsfs}%provides \mathscr

\newtheorem{theorem}{Theorem}[section]
\newtheorem{lemma}[theorem]{Lemma}
\newtheorem{proposition}[theorem]{Proposition}

\newtheorem{Definition}[theorem]{Definition}
\newtheorem{Example}[theorem]{Example}

\newenvironment{example}{\begin{Example}\em}%
        {\hfill$\Box$\smallskip\end{Example}}

\usepackage{color}
        {\hfill$\Box$\smallskip\end{trivlist}\vspace*{-.6cm}}

\def\Maple{\textsc{Maple}}
\def\kk{{\mathbb K}}
 \def\NN{{\mathbb N}} 
 \def\RR{{\mathbb K}} 
\def\DD{{\mathcal D}} \def\ds{\displaystyle}  \def\bm{\boldsymbol}
\def\Bc{{\mathcal B}} 
\def\Qc{{\mathcal Q}} 
\def\II{{\mathcal I}} \def\ee{{\mathcal \varepsilon}}
 
\def\db{{\bm d}} 
\def\pb{{\bm \partial}} 
\def\fb{{\bm f}} 
\def\gb{{\bm g}} 
\def\xb{{\bm x}} 
\def\sint{\begingroup\textstyle\int\endgroup} 
\DeclareMathAlphabet{\mathpzc}{OT1}{pzc}{m}{it}
\def\<{\langle}
\def\>{\rangle}
\def\DDD{{\mathscr D}}
\newcommand{\pp}[1]{\frac{d}{d #1}} 
\newcommand{\id}[1]{\ensuremath{\langle #1 \rangle}}

\newcommand{\dual}[1]{{#1}^{*}}

\title{ 
Deflation and Certified Isolation of Singular 
Zeros of Polynomial Systems
}

\numberofauthors{1} \author{ \alignauthor Angelos
  Mantzaflaris and Bernard Mourrain\\
 \affaddr{GALAAD, INRIA M\'editerran\'ee}\\
 \affaddr{BP 93, 06902 Sophia Antipolis, France}\\
 \email{[FirstName.LastName]@inria.fr}
}

\begin{document}

\maketitle

\begin{abstract}
  We develop a new symbolic-numeric algorithm for the certification of
  singular isolated points, using their associated local ring
  structure and certified numerical computations. An improvement of an
  existing method to compute inverse systems is presented, which
  avoids redundant computation and reduces the size of the
  intermediate linear systems to solve. We derive a one-step
  deflation technique, from the des\-crip\-tion of the multiplicity
  structure in terms of dif\-fe\-ren\-tials. The deflated system can
  be used in Newton-based iterative schemes with quadratic
  convergence.  Starting from a polynomial system and a small-enough
  neighborhood, we obtain a criterion for the existence and uniqueness
  of a singular root of a given multiplicity structure, applying a
  well-chosen symbolic perturbation.  Standard verification methods,
  based eg. on interval arithmetic and a fixed point theorem, are
  employed to certify that there exists a unique perturbed system with
  a singular root in the domain.   Ap\-pli\-ca\-tions to topological
  degree computation and to the analysis of real branches of an
  implicit curve illustrate the method.
\end{abstract}

\vspace{-.3cm}
% [Systems of equations] 
\vspace{-.3cm}
\vspace{-.3cm} 

\section{Introduction}\label{introduction}
A main challenge in algebraic and geometric computing is singular point
identification and treatment. Such problems naturally occur when computing
the topology of implicit curves or surfaces \cite{amw2008}, the intersection of
parametric surfaces in geometric modeling. When algebraic representations are
used, this reduces to solving polynomial systems. Several approaches are
available: algebraic techniques such as  Gr\"obner bases or border bases, resultants, 
subdivision algorithms\cite{mmt09}, \cite{mp09}, homotopies, and so on.
At the end of the day, a numerical
approximation or a box of isolation is usually computed to characterize the
(real) roots of the polynomial system. But we often need to improve the
numerical approximation of the roots. Numerical methods such that Newton's
iteration can be used to improve the quality of the approximation, provided
that we have a simple root. In the presence of a multiple root, the
difficulties are significantly increasing. The numerical approximation can be
of very bad quality, and the methods used to compute this approximation are
converging slowly (or not converging). The situation in practical problems, as encountered in
CAGD for instance, is even worse, since the coefficients of the input
equations are known, with some incertitude. Computing multiple roots of
approximate polynomial systems seems to be an ill-posed problem, since
changing slightly the coefficients may transform a multiple root into
a cluster of simple roots (or even make it disappear).

To tackle this difficult problem, we adopt the following strategy. 
We try to find a (small) perturbation of the input system such that the root
we compute is an exact multiple root of this perturbed system.
In this way, we identify the multiplicity structure and we are able to 
setup deflation techniques which restore the quadratic convergence of the
Newton system. The certification of the multiple root is also possible on the 
symbolically perturbed system by applying a fixed point theorem, based
eg. on interval arithmetic \cite{RG10} or $\alpha$-theorems \cite{ShubSmale93}.

\smallskip\noindent{}\textbf{Related work.}\label{related}
In order to be develop Newton-type methods that converge to multiple roots, 
deflation techniques which consist in adding new
equations in order to reduce the multiplicity have already been
considered. In \cite{Ojika1983463}, by applying a triangulation 
preprocessing step on the Jacobian matrix at the approximate root, minors of
the Jacobian matrix are added to the system to reduce the multiplicity.

In~\cite{Lecerf02}, a presentation of the ideal in a triangular form in a
good position and derivations with respect to the leading variables are used
to iteratively reduce the multiplicity. This process is applied for p-adic
lifting with exact computation.  

In~\cite{lvz06,lvz08}, instead of triangulating the Jacobian matrix, the
number of variables is doubled and new equations are introduced, which are
linear in the new variables. They des\-cribe the kernel of the Jacobian
matrix at the multiple root.

In~\cite{zeng05}, this construction in related to the construction of the
inverse system. The dialytic method of F.S. Macaulay \cite{mac1916} is
revisited for this purpose. 
These deflation methods are applied iteratively until the root becomes
regular, doubling each time the number of variables.

More recent algorithms for the construction of inverse systems are described
eg. in \cite{Marinari:1995:GDM:220346.220368}, reducing the size of the
intermediate linear systems (and exploited in
\cite{Stetter:1996:AZC:236869.236919}), or in \cite{Mourrain97} using an
integration method.

In~\cite{Pope2009606}, a minimization approach is used to reduce the value of
the equations and their derivatives at the approximate root, assuming a basis
of the inverse system is known. 

In \cite{Wu:2008:CMS:1390768.1390812}, the inverse system is constructed via
Macaulay's method; tables of multiplications are deduced and their
eigenvalues are used to improve the approximated root. It is
proved that the convergence is quadratic when the Jacobian has corank 1 at the 
multiple root.

Verification of multiple roots of (approximate) polynomial equations is a
difficult task. The approach proposed in \cite{RG10} consists in introducing
perturbation parameters and to certifying the multiple root of nearby system
by using a fixed point theorem, based on interval arithmetic. It applies
only to cases where the Jacobian has corank equal to 1.

\smallskip\noindent{}\textbf{The univariate case.} %\label{univariate}
In preparation for the multivariate case, we review some techniques
used to treat singular zeros of univariate polynomials, and we present
our method on a univariate instance.

Let $g(x)\in \RR[x]$ be a polynomial which attains at $x=0$ a root of
multiplicity $\mu > 1$. The latter is defined as the smallest positive
integer $\mu$ such that $d^\mu g(0) \ne 0 $ whereas $ g(0)= d
g(0)=\cdots= d^{\mu -1}g(0)=0$.  Here we denote $d^k g(x)=
\frac{d^k}{d x^k} g(x) /{k!} $ the normalized $k-$th order
derivative.%  with respect to $x$.

We see that $\DDD_0= \langle 1, d^1, \dots, d^{\mu -1}
\rangle$ is the maximal space of differentials which is stable under derivation, that vanish when
applied to members of $\Qc_0$, the $\id{x}-$primary component of
$\id{g}$ at $x=0$.
Consider now the symbolically perturbed equation
\begin{equation}\label{eq:univsys}
f_1(x,\bm \ee) = g(x) + \ee_1 + \ee_2 x + \cdots + \ee_{\mu-2} x^{\mu-2}
\end{equation}
and apply every basis element of $\DD_0$ to arrive to the new system
$
\bm f(x,\bm\ee)= \Big(f_1, d_1 f_1, \dots,d^{\mu-1}f_1 \Big)
$
in $\mu -1 $ variables. The $i-$th equation shall be $ f_i= d^{i-1}
f_1 = d^{i-1} g + \sum_{k=i}^{\mu-2} x^{k-i+1}\ee_k $, i.e linear in
$\bm \ee$, the last one being $f_{\mu}= d^{\mu-1} g(x)$.
% We set $\ee_{\mu-1}=0$. 
This system deflates the root, as we see that the
determinant of it's Jacobian matrix at $(0,\bm 0)$ is
$$
\det J_{\bm f}(0,\bm 0)=
\left|
\begin{array}{c|c}
\begin{array}{c}
\pp{x} f_1 \\
\vdots\\
\pp{x} f_{\mu-1}
\end{array} & 
\begin{array}{ccc}
1& & 0 \\
& \ddots & \\
0 & & 1 
\end{array}
\\ \hline
\pp{x}f_{\mu} & 0 
\end{array}
\right| \begin{array}{l}
\ \\
= - \mu   d f_\mu(0,\bm 0)\\
= - \mu  d^\mu g(0)\ne 0 .
\end{array}
$$
Now suppose that $\zeta^*$ is an approximate zero, close to $x=\zeta$. We
can still compute $\DD_{\zeta}$ by evaluating $g(x)$ and the derivatives
up to a threshold relative to the error in $\zeta^*$. Then we can
form~\eqref{eq:univsys} and use verification techniques to certify the
root. Checking that the Newton operator is contracting shows the existence
and unicity of a multiple root in a neighborhood of the input data. 
 We are going to extend this approach, described in \cite{RG10},
to multi-dimensional isolated multiple roots.

\smallskip\noindent{}\textbf{Our approach.} It consists of the following steps:

(a) Compute a basis for the dual space and of the local quotient ring at a
 given (approximate) singular point.

(b) Deflate the system by augmenting it with new equations derived from the
 dual basis, introducing adequate perturbation terms.

(c) Certify the singular point and its multiplicity structure for the perturbed
  system checking the contraction property of Newton iteration (eg. via interval arithmetic).
\smallskip

In step (a), a dual basis at the singular point is computed by means of
linear algebra, based on the integration approach of \cite{Mourrain97}. We describe an
improvement of this method, which yields directly a triangular dual basis with no
redundant computation. This method has the advantage to reduce significantly
the size of the linear systems to solve at each step, compared to Macaulay's
type methods \cite{mac1916,lvz06,lvz08,zeng05}.  In the case of an approximate singular point, errors
are introduced in the coefficients of the basis elements. Yet a successful
computation is feasible. In particular, the support of the basis elements is
revealed by this approximate process. 

In the deflation step (b), new equations and new variables are introduced in
order to arrive to a new polynomial system where the singularity is
obviated. The new variables correspond to perturbations of the initial
equations along specific polynomials, which form a dual counterpart to the basis of
the dual space. One of the deflated systems that we compute from the dual system is a
square $n \times n$ system with a simple root. This improves the deflation
techniques described in \cite{lvz06,lvz08,zeng05}, which require additional
variables and possibly several deflation steps.  New variables are introduced
only in the case where we want to certify the multiplicity structure.  The
perturbation techniques that we use extend the approach of \cite{RG10} to
general cases where the co-rank of the Jacobian matrix could be bigger than
one.
The verification step (c) is mostly a contraction condition, using
eg. techniques as in \cite{RG10}. This step acts on the (approximate) deflated system, since
verifying a simple solution of the deflated system induces a
certificate of an exact singular point of (a nearby to) the initial system.

We are going to detail the different steps in the following sections,
starting with notations in Sect. \ref{results}, dual basis in
Sect.~\ref{sec:dual}, deflation in Sect.~\ref{sec:deflate}, and
certification in Sect.~\ref{sec:verify}.  In the last section, we will show
examples and applications to the topology analysis of curves.

%%%%%%%%%%%%%%%%%%%%%%%%%%%%%%%%%%%%%%%%%%5

\section{Preliminaries and main results}\label{results}

We denote by $R= \RR[\bm x]$ a polynomial ring over the field $\RR$ of
characteristic zero.  Also, the \emph{dual ring} $\dual{R}$ is the
space of linear functionals $\Lambda: R \to \RR$. It is commonly
identified as the space of formal series $\RR[[\pb]]$ where
$\pb=(\partial_1,\dots, \partial_n)$ are formal variables. 
Thus we view dual elements as formal series in differential
operators at a point $\bm \zeta\in\RR^{n}$. 
 To specify that we use the point
$\bm\zeta$, we also denote these differentials $\pb_{\bm\zeta}$. When applying $\Lambda(\pb_{\bm \zeta})\in
\RR[[\pb_{\bm \zeta}]]$ to a polynomial $g(\bm x)\in R$ we will denote by
 $\Lambda^{\bm \zeta}[g]=
 \Lambda^{\bm \zeta} g= 
 \Lambda(\pb_{\bm\zeta})[g(\xb)]$ the operation
\begin{align}\label{eq:defdual}
  \Lambda^{\bm \zeta}[g]= \sum_{\bm \alpha\in\NN^n}
  \frac{ \lambda_{\bm \alpha}}{ \alpha_1!\cdots \alpha_n! } \cdot
  \frac{d^{|\bm \alpha|} g}{d_1^{\alpha_1}\cdots d_n^{\alpha_n}} (\bm \zeta) ,
\end{align}
for $\Lambda(\pb_{\zeta})= \sum \lambda_{\bm \alpha} \bm {1\over \bm\alpha !}\pb_{\zeta}^{\bm \alpha} \in
\RR[[\pb_{\zeta}]]$. 
Extending this definition to an ordered set $\DD=(\Lambda_1,\dots,\Lambda_\mu)\in\kk[[\pb]]^{\mu}$, we
shall denote $\DD^{\bm \zeta}[g] = \bm(\Lambda_1^{\bm \zeta}g,\dots,\Lambda_\mu^{\bm \zeta}
g) $.
In some cases, it is convenient to use normalized
differentials instead of $\pb$: for any $\alpha\in \NN^{n}$, we denote 
${\db}_{\bm \zeta}^{\bm \alpha}={1\over {\bm \alpha} !} \,\pb_{\bm \zeta}^{\bm \alpha}
$.
When $\bm\zeta=\bm 0$, we have $\db_{\bm 0}^{\bm \alpha} \bm x^{\bm \beta} =
1$ if $\bm \alpha=\bm \beta$ and $0$ otherwise.
More generally, $(\db_{\bm\zeta}^{\bm\alpha})_{\bm\alpha\in \NN^n}$ is the dual basis of
$((\xb-\bm\zeta)^{\bm\alpha})_{\bm\alpha\in \NN^n}$.

For $\Lambda\in R^{*}$ and $p\in R$, let $p\cdot \Lambda: q\mapsto
\Lambda(p\,q)$. We check that 
\begin{equation}\label{eq:derive}
  (x_{i}-\zeta_{i})\cdot {\bm\partial}_{\zeta}^{\bm \alpha} = \pp{\partial_{i,\zeta}} ({\bm
  \partial}_{\zeta}^{\bm \alpha}).
\end{equation}
This property shall be useful in the sequel.

\subsection{Isolated points and differentials}
Let $\id{\bm f}= \id{f_1,\dots,f_s}$ be an ideal of $R$, $\bm
\zeta\in \RR$ a root of $\bm f$ and $m_{\bm
  \zeta}=\id{x_1-\zeta_1,\dots,x_n-\zeta_n}$ the maximal ideal at
$\bm \zeta$. Suppose that $\bm \zeta$ is an isolated root of $\bm f$, then 
a minimal primary  decomposition of 
$\ds \II
= \bigcap_{ \mathcal Q\, \mathrm{prim.}\supset \II} \mathcal Q
$
contains a primary component $\mathcal Q_{\bm\zeta}$ such that
$\sqrt{\Qc_{\bm \zeta}}=m_{\bm  \zeta}$ and $\sqrt{\Qc'}
\not\subset m_{\bm  \zeta}$ for the other primary components $\mathcal Q'$
associated to $\II$ \cite{AtMa69}.

As $\sqrt{\Qc_{\zeta}}=m_{\bm \zeta}$, $R/\Qc_{\bm \zeta}$ is a finite
dimensional vector space. The multiplicity $\mu_{\bm\zeta}$ of
$\bm\zeta$ is defined as the dimension of $R/\Qc_{\bm \zeta}$.  A
point of multiplicity one is called regular point, or simple root,
otherwise we say that $\bm \zeta$ is a singular isolated point, or
multiple root of $\bm f$.  In the latter case we have %a Jacobian
$J_{\bm f}({\bm\zeta})= 0$.

We can now define the dual space of an isolated point.
\begin{Definition}\label{def:ds}
The \emph{dual space} of $\II$ is the subspace of elements of $\RR[[\pb_{\zeta}]]$ 
that vanish on all the elements of $\II$. It is also called the
orthogonal of $\II$ and denoted by $\II^{\perp}$.
\end{Definition}

Consider now the \emph{orthogonal} of $\Qc_{\bm \zeta}$, i.e. the
subspace $\DDD_{\bm \zeta}$ of elements of $\dual R$ that vanish on
members of $\Qc_{\bm \zeta}$, namely
$$
\Qc_{\bm \zeta}^{\perp}= \DDD_{\bm \zeta} =\{\Lambda\in\dual R\, : \, \Lambda^{\zeta} [p]=0,\, \forall p\in\Qc_{\zeta} \}.
$$
The following lemma is an essential property that allows extraction of
the local structure $\DDD_{\bm \zeta}$ directly from the ``global'' ideal
$\II=\id{\bm f}$, notably by matrix methods outlined in Sect.~\ref{sec:dual}.
\begin{proposition}[{\cite[Th.~8]{Mourrain97}}]\label{prop:diffpol}
  For any isolated point $\bm \zeta\in\RR$ of $\bm f$, we have 
$\id{\bm f}^{\perp} \cap \RR[\pb_{\bm \zeta}]=\DDD_{\bm \zeta}$ .
\end{proposition}
In other words, we can identify
$\DDD_{\bm\zeta}=\Qc_{\bm\zeta}^{\perp}$ with the space of polynomial
differential operators that vanish at $\bm\zeta$ on every element
of $\II$.

The space $\DDD_{\bm\zeta}$ is a vector space of polynomials in $\pb_{\bm\zeta}$
dimension $\mu_{\bm\zeta}$, the multiplicity of $\bm\zeta$. As the
variables $(x_{i}-\zeta_{i})$ act on $R^{*}$ as derivations
(see \eqref{eq:derive}),
$\DDD_{\bm\zeta}$ is a space of differential polynomials in $\bm\partial_{\zeta}$, which is stable by
derivation. This property will be used explicitly in constructing $\DDD_{\bm\zeta}$  (Sec. \ref{sec:dual}).

\begin{Definition}
The \emph{nilindex} of $\Qc_{\bm \zeta}$ is the maximal
integer $N\in\NN$ s.t. $m_\zeta^N \not\subset \Qc_{\bm \zeta}$.  
\end{Definition}
It is directly seen that the maximal degree of an element of $\DDD_{\bm\zeta}$
has to be equal to $N$, also known as the \emph{depth} of $\DDD_{\bm\zeta}$.

\subsection{Quotient ring and dual structure}
In this section we explore the relation between the dual ring and the
quotient $R/\Qc_{\bm\zeta}$ where $\Qc_{\bm\zeta}$ is the primary component of the
isolated point $\bm\zeta$. We show how to extract a basis of this quotient ring
from the support of the elements of $\DDD_{\bm\zeta}$ and how $\DDD_{\bm\zeta}$ can be used to
reduce any polynomial modulo $\Qc_{\bm \zeta}$.

It is convenient in terms of notation to make the assumption $\bm
\zeta=\bm 0$. This saves some indices, while it poses no constraint (since
it implies a linear change of coordinates), and shall be adopted hereafter
and in the next section.

Let $\text{supp} \DDD_{\bm 0}$ be the set of exponents of monomials appearing
in $\DDD_{\bm 0}$, with a non-zero coefficient. These are of degree at most
$N$, the nilindex of $\Qc_{\bm 0}$.
Since $\forall \Lambda \in \DDD_{\bm 0}$, $\Lambda^{\bm 0} [p] =0$ iff $p\in\Qc_{\bm 0}$, we
derive that $\text{supp}\, \DDD_{\bm 0}= \{\bm \alpha \, : \, \bm x^{\bm
  \alpha}\notin \Qc_{\bm 0} \}$.  In particular, we can find a basis of
$R/\Qc_{\bm 0}$ between the monomials $\{ \bm x^{\bm \alpha}\, :\, \bm
\alpha\in\text{supp}\, \DDD \} $. 
This is a finite set of monomials, since their degree is bounded by the
nilindex of $\Qc_{\bm 0}$. 

Given a standard basis $\Bc= (\bm x^{\bm \beta_i})_{i=1,\dots,\mu}$ of
$R/ \Qc_{\bm 0}$ and, for all monomials $x^{\bm \gamma_j}\notin \Qc_{\bm 0},\,
j=1,\dots, s-\mu$, $s=\#\text{supp}\, \DDD_{\bm 0}$, with $x^{\bm
  \gamma_j}\notin \Bc$, the expression (normal form) of
\begin{align}
  x^{\bm \gamma_j} = \sum_{i=1}^{\mu} \lambda_{ij} \bm x^{\bm \beta_i} \mod \Qc_{\bm 0} 
\end{align}
in the basis $\Bc$ then the dual elements~\cite[Prop.~13]{Mourrain97}
\begin{align}\label{cbasis}
\Lambda_i =  \db^{\bm \beta_i} + \sum_{j=1}^{s-\mu} \lambda_{ij} \db^{\bm \gamma_j} ,
\end{align}
for $i=1,\dots, \mu$ form a basis of $\DDD_{\bm 0}$.
We give a proof of this fact in the following lemma.
\begin{lemma}\label{lem:nf}
  The set of elements $\DD=(\Lambda_i)_{i=1,\dots,\mu}$ is a basis of
  $\DDD_{\bm 0}$ and the normal form of any $g(\bm x)\in R$ with respect to
  the standard basis $\Bc=(\xb^{\bm \beta_i} )_{i=1,\dots,\mu}$ is 
  \begin{align}\label{eq:normalf}
  \text{NF}(g)= \sum_{i=1}^{\mu} \Lambda_i^{\bm 0}[ g]\, \xb^{\bm \beta_i} .
  \end{align}
\end{lemma}
\begin{proof}
  First note that the elements of $\DD$ are linearly independent.
  Now, by construction, $\sum_{i=1}^\mu \Lambda_i^{\bm 0}[\xb^{\bm
    \alpha}]=\text{NF}(\xb^{\bm \alpha})$ for all $\xb^{\bm
    \alpha}\notin \Qc_{\bm 0}$, eg. $\text{NF}(\xb^{\bm \beta_i})=\xb^{\bm \beta_i}$.
  Also, for $\xb^{\bm \alpha} \in \Qc_{\bm 0}$,
  $\forall i, \, \Lambda_i^{\bm 0}(\xb^{\bm \alpha})=0$, since $\bm
  \alpha\notin \text{supp}\, \DD$. Thus the elements of $\DD$ compute
  $\text{NF}(\bm \cdot)$ on all monomials of $R$,
  and~(\ref{eq:normalf}) follows by li\-nea\-ri\-ty. We deduce that
  $\DD$ is a basis of $\DDD_{\bm 0}$, as in Def.~\ref{def:ds}.  \if 0
  Now for arbitrary $g(\bm x)\in R$ we write $g(x)= q(x) + NF(g)$,
  with $q(x)\in \Qc_{0}$ and $NF(g)\in R/\Qc_{\bm 0}$, say $NF(g)=
  \sum_{i=1}^{\mu}c_i \xb^{\beta_{i}}$.  Then we have
\begin{align}
  \sum_{i=1}^\mu \Lambda_i [g] &= 
    \sum_{i=1}^\mu \Lambda_i [q + NF(g)] &= \\
    \sum_{i=1}^\mu \Lambda_i [NF(g)] &= 
    \sum_{i=1}^\mu c_i \Lambda_i [\xb^{\beta_{i}}] = 
    NF(g)  .
\end{align}
\fi
\end{proof}
Computing the normal form of the border monomials of $\Bc$ via
\eqref{eq:normalf} also yields the border basis relations and the ope\-ra\-tors
of multiplication in the quotient $R/\Qc_{\bm 0}$ (see eg. \cite{irsea-em-07} for more properties).

If a graded monomial ordering is fixed and $\Bc=(\xb^{\bm \beta_i}
)_{i=1,..,\mu}$ is the corresponding standard basis of $R/\Qc_{\bm
  0}$, then $\bm d^{\bm \beta_i}$ is the leading term of \eqref{cbasis}
wrt the opposite ordering~\cite[Th.~3.1]{lvz08}.

Conversely, if we are given a basis $\DD$ of $\DDD_{\bm 0}$ whose
coefficient matrix in the dual monomials basis $(\db^{\bm \alpha})_{\bm
  \alpha\notin \Qc_{\bm 0} }$ is $D\in \RR^{\mu\times s}$, we can compute a
basis of $R/\Qc_{\bm 0}$ by choosing $\mu$
independent columns of $D$, say those indexed by ${\db}^{\bm
  \beta_i},\, i=1,\dots,\mu$ . If $G\in\RR^{\mu\times\mu}$ is the
(invertible) matrix formed by these columns, then $D':= G^{-1}D$, is
\begin{align}\label{eq:diagonalbasis}
D' = 
\bordermatrix[{[]}]{ 
           & \bm \beta_1&\cdots & \bm \beta_{\mu}  & \bm\gamma_{1}   &\cdots & \bm\gamma_{s-\mu}  \cr
\Lambda_1'  & 1          &       & 0               & \lambda_{1,1}&\cdots & \lambda_{1,s-\mu}\cr
\ \vdots   &            &\ddots &                 & \vdots      &       & \vdots         \cr
\Lambda_\mu' & 0          &       & 1               &\lambda_{\mu,1}&\cdots & \lambda_{\mu,s-\mu}  
  } ,
\end{align}
i.e. a basis of the form~\eqref{cbasis}.
Note that an arbitrary basis of $\DDD$ does not have the above
diagonal form, nor does it directly provide a basis for $R/ \Qc_{\bm 0}$.

For $t\in \NN$, $\DDD_{t}$ denotes the vector space of polynomials of
$\DDD$ of degree $\le t$. The Hilbert function $h:\NN \to \NN $ is
defined by $h(t)=\dim( \DDD_t),\, t\ge 0$, hence $h(0)=1$ and
$h(t)=\dim \DDD$ for $t\ge N$. The integer $h(1)-1=\text{corank}\,
J_{\bm f}$ is known as the \emph{breadth} of $\DDD$.

\section{\hspace{-.4cm}Computing local ring structure}\label{sec:dual}

The computation of a local basis, given a system and a point, is done
essentially by matrix-kernel computations, and consequently it can be
carried out numerically, even when the point or even the system is
inexact.  Throughout the section we suppose $\bm f\in R^m$ and $\bm\zeta\in
\RR^n$ with $\bm f(\bm \zeta)=0$.

Several matrix constructions have been proposed, that use different
conditions to identify the dual space as a null-space.
They are based on the \emph{stability property} of the dual basis:
\begin{align}\label{eq:stability}
\forall\, \Lambda\in\DDD_t , \ \ \pp{\partial_i}\Lambda \in\DDD_{t-1}  \ ,\ \ \ 
i=1,\dots,n  .
\end{align}
We list existing algorithms that compute dual-space bases:

\textbullet\ 
As pointed out in~\eqref{eq:derive}, an equivalent form
of~\eqref{eq:stability} is: $ \forall \Lambda\in\DDD_t , \Lambda [ x_i f_j] = \bm 0 $, $\forall i,j=1,\dots,n $. Macaulay's
method~\cite{mac1916}
uses it
to derive the algorithm that is
outlined in Sect.~\ref{sec:macaulay}. %Final step: 12x10

\textbullet\ 
  In \cite{Marinari:1995:GDM:220346.220368} they exploit~\eqref{eq:stability} by
  forming the matrix $D_i$ of the map $\ds \pp{\partial_i} \, :\, \RR[\bm \partial]_{t} \to \RR[\bm
  \partial]_{t-1} $ for all $i=1,\dots,n$ and some triangular decomposition of the
  differential polynomials in terms the differential variables.  This
  approach was used in \cite{Stetter:1996:AZC:236869.236919} to reduce the
  row dimension of Macaulay's matrix, but not the column dimension.
The closedness condition is also used in~\cite{zeng09} to identify a superset of $\text{supp}\,
  \DDD_{t+1}$. %Final step: 5x4

\textbullet\ 
  The \textit{integration method} in~\cite{Mourrain97} ``integrates'' elements of a
  basis of $\DDD_{t}$, and obtains \emph{a priori} knowledge of the form
  of elements in degree ${t+1}$ (Sect.~\ref{sec:integration}). %Final step: 4x6

All methods are incremental, in the sense that they start by setting
$\DD_0=(\bm 1)$ and continue by computing $\DD_i,\
i=1,\dots,N,N+1$. When $\#\DD_N=\#\DD_{N+1}$ then $\DD_{N}$ is a basis
of $\DDD$, and $N$ is the nilindex of $\Qc$.

We shall review two of these approaches to compute a basis for $\DDD$, and
then describe an improvement, that allows simultaneous computation of a
quotient ring basis while avoiding redundant computations.

\subsection{Macaulay's dialytic matrices}\label{sec:macaulay}

This matrix construction is presented in~\cite[Ch.~4]{mac1916}, a
mo\-dern introduction is contained in~\cite{zeng05}, together with an
implementation of the method in
ApaTools\footnote{http://www.neiu.edu/~zzeng/apatools.htm}.

The idea behind the algorithm is the following: An e\-le\-ment of $\DDD$
is of the form
%\begin{align}
$\ds
  \Lambda(\pb) = \sum_{|\bm \alpha|\le N} \lambda_{\bm \alpha} \db^{\bm \alpha}
$
%\end{align}
under the condition: $\Lambda^{\bm 0}$ evaluates to $0$ at any $g\in\id{\bm f}$, i.e.
$\Lambda^{\bm 0}(g)=\Lambda^{\bm 0}(\sum g_i f_i)=0 \iff \Lambda^{\bm 0}(x^{\bm\beta} f_i)=0$
for all monomials $x^{\bm\beta}$.
If we apply this condition recursively for $|\bm \alpha|\le N$ we get
a vector of coefficients $(\lambda_{\bm \alpha})_{{|\bm \alpha|\le
    N}}$ in the (right) kernel of the matrix with rows indexed by
constraints $ \Lambda^{\bm 0} [ \bm x^{\bm \beta} \bm f_i] = 0 $,
$|\bm \beta|\le N-1$.

Note that the only requirement is to be able to perform derivation of
the input equations and evaluation at $\bm \zeta= \bm 0$.

\begin{example} \label{ex:macaulay}
  Let $f_1= x_1 - x_2 + x_1^2$, $f_2= x_1 - x_2 + x_2^2$.  We also
  refer the reader to \cite[Ex.~2]{zeng05} for a detailed
  demonstration of Macaulay's method on the same instance.  The
  matrices in order $1$ and $2$ are:
{\small
  \[
  \bordermatrix[{[]}]{ & 1 & d_1 & d_2 \cr f_1 & 0&1&-1 \cr f_2 &
    0&1&-1 } \ \ , \ \ \bordermatrix[{[]}]{ & 1 & d_1 & d_2 & d_1^2 &
    d_1d_2 & d_2^2 \cr \ f_1 & 0&1&-1&1&0&0\cr \ f_2 & 0&1&-1&0&0&1\cr
    x_1f_1 & 0&0&0&1&-1&0\cr x_1f_2 & 0&0&0&1&-1&0\cr x_2f_1 &
    0&0&0&0&1&-1\cr x_2f_2 & 0&0&0&0&1&-1 } .
  \]
}
The kernel of the left matrix gives $\DD_1=(1,d_1+d_2)$. Expanding up
to order two, we get the matrix on the right, and $\DD_2=(1,d_1+d_2,
-d_1+d_1^2+d_1d_2+d_2^2)$. If we expand up to depth $3$ we get the
same null-space, thus $\DD=\DD_2$.
\end{example}

\subsection{Integration method}\label{sec:integration}

This method is presented in~\cite{Mourrain97}.
It is an evolution of Macaulay's method, in the sense that the
matrices are not indexed by all differentials, but just by elements
based on knowledge of the previous step. This performs a computation
adapted to the given input and results in smaller matrices.

For $\Lambda\in\RR[\bm \partial]$, we denote by $\sint_k \Lambda$ the element
$\Phi\in\RR[\bm \partial]$ with the
property $\pp{\partial_k}\Phi(\bm \partial)=\Lambda(\bm \partial)$ and with no constant
term wrt $\partial_k$.
\begin{theorem}[ {\cite[Th.~15]{Mourrain97}} ]\label{th:integrate}
  Let $\id{\Lambda_1,\Lambda_2, \dots,\Lambda_s}$ be a
  basis of $\DDD_{t-1}$, that is, the subspace of $\DDD$ of elements of
  order at most $t-1$.  
  An element $\Lambda\in \RR[\bm \partial]$ with no constant term lies in
  $\DDD_t$ iff it is of the form:
\begin{equation}\label{eq:candidate}
\Lambda(\bm \partial)= 
\sum_{i=1}^s  
\sum_{k=1}^n
\lambda_{ik} \sint_k \Lambda_{i}(\partial_1,\dots,\partial_k,0,\dots,0), 
\end{equation}
for $\lambda_{ik}\in\RR$, and the following two conditions hold:
\begin{itemize}
\item[(i)] 
  $\displaystyle \sum_{i=1}^s \lambda_{ik} \pp{\partial_l}
  \Lambda_{i}(\bm \partial) - \sum_{i=1}^s \lambda_{il} \pp{\partial_k}
  \Lambda_{i}(\bm \partial) =0$,\\ for all $1\le k\le l\le n$ .
\item[(ii)]
$\Lambda^{\bm \zeta} [f_k] =0$, for $k=1,\dots,n$ .
\end{itemize}
\end{theorem}
Condition $(i)$ is equivalent to $\pp{\partial_k}\Lambda \in\DDD_{t-1}$, for $1
\le k\le n$. Thus the two conditions express exactly the fact that
$\DDD$ must be stable by derivation and his members must vanish on
$\id{\bm f}$.

This gives the following algorithm to compute the dual basis:
Start with $\DD_0=\langle 1 \rangle$. Given a basis of $\DDD_{t-1}$ we
generate the $ns$ candidate elements $\sint_k
\Lambda_{i-1} (\partial_1,\dots,\partial_k,0,\dots,0)$. Conditions $(i)$ and
$(ii)$ give a linear system with unknowns $\lambda_{ik}$. The columns
of the corresponding matrix are indexed by the candidate elements.
Then the kernel of the matrix gives the possible solutions, which we
add to $\DD_{t-1}$ to obtain $\DD_t$.
If for some $t$ there are no further new elements, then $\DD=\DD_t$ is a basis of $\DDD$.

\begin{example}\label{ex:integration}
  Consider the instance of Ex.~\ref{ex:macaulay}. We have
  $f_1(\bm \zeta)=f_2(\bm \zeta)=0$, thus we
  set $\DD_0=\{1\}$. Equation~(\ref{eq:candidate}) gives $\Lambda=
  \lambda_1 d_{1} + \lambda_2 d_{2}$. Condition (i) induces no
  constraints and (ii) yields the system
{\small
\begin{align}\label{ex:integration:1}
\left[\begin{array}{cc}
1 & -1 \\
1 & -1 
\end{array}
\right]
\left[ \begin{array}{cc}\lambda_1\\ \lambda_2 \end{array}\right]
=0
\end{align}
}
where the columns are indexed by $d_1,d_2$.  We get
$\lambda_1=\lambda_2=1$ from the kernel of this matrix, thus
$\DD_1=\{1,d_1+d_2 \}$.

For the second step, we compute the elements of $\DD_2$, that must be
of the form
$
\Lambda= \lambda_1 d_1 + \lambda_2 d_2 + \lambda_3 d_1^2 + \lambda_4 (d_1d_2+d_2^2)$.
Condition (i) yields $\lambda_3 - \lambda_4 = 0$, and together with
(ii) we form the system
{\small
\begin{align}\label{ex:integration:2}
\left[\begin{array}{cccc}
0 &  0 & 1 & -1 \\
1 & -1 & 1 & 0 \\
1 & -1 & 0 & 1 
\end{array}
\right]
\left[ \begin{array}{cc}\lambda_1\\ \vdots \\ \lambda_4  \end{array}\right]
=0   ,
\end{align}
}
with columns indexed by $d_1, d_1^2, d_2, d_1d_2+d_2^2$.
We get two vectors in the kernel, the first yielding again $d_1+d_2$
and a second one for $\lambda_1=-1, \lambda_2=0, \lambda_3=
\lambda_4= 1$, so we deduce that $-d_1+d_1^2+d_1d_2+d_2^2$ is a new
element of $\DD_2$.

In the third step we have
\begin{align*}
\Lambda=& \lambda_1 d_1 + \lambda_2 d_2 + \lambda_3 d_1^2 +  \lambda_4(d_1d_2+d_2^2) +  \\
& \lambda_5 (d_1^3-d_1^2) + \lambda_6 (d_2^3+d_1d_2^2+d_1^2d_2 - d_1d_2) ,
\end{align*}
condition (i) leads to 
$\lambda_3 - \lambda_4 + (\lambda_5-\lambda_6)d_1+(\lambda_5-\lambda_6)d_2 = 0
$, and together with condition (ii) we arrive to 
{\small
\begin{align}\label{ex:integration:3}
\left[\begin{array}{cccccc}
0 & 0  & 0  & 0  & 1  & -1 \\
0 & 0  & 1  & -1  & 0 & 0 \\
1 & -1 & 1 & 0 & -1  & 0 \\
1 & -1 & 0  & 1 & 0 & 0 
\end{array}
\right]
\left[ \begin{array}{cc}\lambda_1\\ \vdots \\  \lambda_6  \end{array}\right]
=0 .
\end{align}
}
Since the kernel of this matrix gives elements that are already in
$\DD_2$, we derive that $\DD=\DD_2=\DD_3$ and the algorithm terminates.
%\DD stabilizes

Note that for this example Macaulay's method ends with a matrix of
size $12\times 10$, instead of $4\times 6$ in this
approach.
\end{example}

\subsection{Computing a primal-dual pair}

In this section we provide a process that allows simultaneous
computation of a basis pair $(\DD,\Bc)$ of $\DDD$ and $R/\Qc$.

Computing a basis of $\DDD$ degree by degree involves duplicated
computations. The successive spaces computed are $\DDD_1\subset \cdots
\subset \DDD_N = \DDD_{N+1}$.
It is more efficient to compute only new elements $\Lambda \in
\DDD_t $ which are independent in $\DDD_t / \DDD_{t-1}$ at step $t$. 

Also, once dual basis is computed, one has to transform it into the
form~\eqref{cbasis}, in order to identify a basis of $R/\Qc$ as
well. This transformation can be done \emph{a posteriori}, by finding
a sub-matrix of full rank and then performing Gauss-Jordan elimination
over this sub-matrix, to reach matrix
form~\eqref{eq:diagonalbasis}.

We introduce a condition (iii) extending Th.~\ref{th:integrate}, that
addresses these two issues: It allows the computation of a total of
$\mu$ independent elements throughout execution, and returns a
``triangular basis'', e.g. a basis of $R /\Qc$ is identified.

\begin{lemma}\label{lem:tbasis}
  Let $\DD_{t-1}=(\Lambda_1,\dots,\Lambda_k)$ be a basis of
  $\DDD_{t-1}$, whose coefficient matrix is 
{\small
\begin{align}\label{eq:trbasis}
\bordermatrix[{[]}]{ 
           & \bm \beta_1&\cdots & \bm \beta_{k}  & \gamma_{1}   &\cdots & \gamma_{s-k}  \cr
\Lambda_1  & 1          &   *   &     *           &  *     &\cdots & * \cr
\ \vdots   &       0    &\ddots &     *           & \vdots      &       & \vdots         \cr
\Lambda_k &  0      &  0    & 1               &    *      &\cdots &    *
  },
\end{align}
}
yielding the standard basis
  $\Bc=(\xb^{\bm\beta_{i}})_{i=1,\ldots,k}$.
 An element $\Lambda \in \RR[\bm \partial]$ is not zero 
in $\DDD_t/\DDD_{t-1}$ iff in addition to (i), (ii) of Th.~\ref{th:integrate} we impose:
  \begin{itemize}
  \item[(iii)] $\Lambda [\xb^{\bm \beta_i}] = 0$,\ $1\le i \le k $ .
\end{itemize}
\end{lemma}
\begin{proof}
  Let $\Lambda\in \RR[\bm \partial]$ be a non-zero functional satisfying
  (i), (ii) and (iii). Then $\Lambda\in \DDD_{t}$
  and $\Lambda [\xb^{\bm \beta_i}]= 0$ for $i=1,\ldots,k$. If $\Lambda\in
  \DDD_{t-1}$, then $\Lambda = \sum_{i=1}^{k} \, \lambda_{i}\, \Lambda_{i}$.
Take for $i_{0}$ the minimal $i$ such that $\lambda_{i}\neq 0$. Then $\Lambda
[\xb^{\bm\beta_{i_{0}}}]=\lambda_{i_{0}}$, which is in contradiction with condition (iii).
Thus, the non-zero solutions of (i), (ii) and (iii) correspond to the
elements which are not zero in $\DDD_t/\DDD_{t-1}$.
\end{proof}

The above constraint is easy to realize; it is equivalent to $\forall
i,\, \db_{\bm \zeta}^{\bm \beta_i}\notin \text{supp}\, \Lambda^{\bm
  \zeta}$, which implies adding a row (linear constraint) for every
$i$.  In many cases this constraint is just $\lambda_{ik}=0$ for some
$i,k$, thus we rather remove the column corresponding to
$\lambda_{ik}$ instead of adding a row. Either way, this lemma allows
to shrink the kernel of the matrix and compute only new dual elements.

Let us explore our running example, to demonstrate the essence of this
improvement.
\begin{example}\label{ex:improve}
We re-run Ex.~\ref{ex:integration} using Lem.~\ref{lem:tbasis}.

In the initialization step $\DD_{0}=(1)$ is already in triangular form
with respect to $\Bc_{0}=\{1\}$. For the first step, we demand
$\Lambda [1]=0$, thus the matrix is the same as
~\eqref{ex:integration:1}, yielding $\DD_1=(1,d_1+d_2)$.  We extend
$\Bc_{1}=\{1,x_{2}\}$, so that $\DD_1$ is triangular with respect to
$\Bc_{1}$.

In the second step we remove from~\eqref{ex:integration:2} the second
co\-lu\-mn, hence
{\small
\begin{align*}
\left[\begin{array}{ccc}
0 &  1 & -1 \\
1 &  1 & 0 \\
1 &  0 & 1 
\end{array}
\right]
\left[ \begin{array}{cc}\lambda_1\\ \lambda_3 \\ \lambda_4  \end{array}\right]
=0  ,
\end{align*}
}
yielding a single solution $-d_1+d_1^2+d_1d_2+d_2^2$. We extend $\Bc_{1}$ by
adding monomial $x_{1}$: $\Bc_{1}=\{1,x_{2},x_{1}\}$. 

For the final step, we search an element with $\Lambda [x_1] = \Lambda [x_2] = 0$ thus \eqref{ex:integration:3} loses two columns:
{\small
\begin{align*}
\left[\begin{array}{cccc}\label{ex:integration:2}
 0  & 0  & 1  & -1 \\
 1  & -1 & 0  & 0 \\
 1  & 0  & -1 & 0 \\
 0  & 1  & 0  & 0 
\end{array}
\right]
\left[ \begin{array}{cc}\lambda_3\\ \vdots \\  \lambda_6  \end{array}\right]
=0 .
\end{align*}
}
We find an empty kernel, thus we recover the triangular basis
$\DD=\DD_2$, which is then diagonalized to reach the form:
{\small
\[
\bordermatrix[{[]}]{ 
           &  1  & d_2  & d_1 & d_1^2 & d_1d_2 & d_2^2  \cr
\Lambda_1  &  1  & 0    & 0   & 0    & 0      & 0     \cr
\Lambda_2  &  0  & 1    & 0   & 1    & 1      & 1     \cr
\Lambda_3  &  0  & 0    & 1   &-1    & -1     & -1
  } .
\]
}
This diagonal basis is dual to the basis $\Bc=(1, x_2, x_1)$ of the
quotient ring and also provides a normal form algorithm
(Lem.~\ref{lem:nf}) wrt $\Bc$.
In the final step we generated a $4\times 4$ matrix, size smaller
compared to all previous methods.
\end{example}
This technique for computing $\Bc$ can be applied similarly 
to other the matrix methods, e.g. Macaulay's dialytic method.

If $h(t)-h(t-1)> 1$, ie. there are more than one elements in step $t$,
then the choice of monomials to add to $\Bc$ is obtained by extracting a non-zero
 maximal minor from the coefficient matrix in $(\bm d^{\bm \alpha})$.
In practice, we will look first at the monomials of smallest
degree.

\subsection{Approximate dual basis}
In our deflation method, we assume that the multiple point is known
approximately and we use implicitly Taylor's expansion of the
polynomials at this approximate point to deduce the dual basis,
applying the algorithm of the previous section. To handle safely the
numerical problems which may occur, we utilize the following techniques:

$\bullet$ At each step, the solutions of linear system~(\ref{eq:candidate}, i-iii)
are computed via Singular Value Decomposition.
Using a given threshold, we determine the numerical rank and an orthogonal
basis of the solutions from the last singular values and the last columns of
the right factor of the SVD.

$\bullet$ For the computation of the monomials which define the equations~(\ref{lem:tbasis}, iii)
 at the next step, we apply QR decomposition on the transpose of
 the basis to extract a non-zero maximal minor. The monomials indexing this
 minor are used to determine constraints~(\ref{eq:candidate}, i-iii).
A similar numerical technique is employed in \cite{zeng09}, for Macaulay's method.

\section{Deflation of a  singular point}\label{sec:deflate}
We consider a system of equations $\bm f=(f_{1},\ldots,f_{s})$, which has
a multiple point at $\xb=\bm\zeta$. Also, let
$\Bc=(b_{1},\ldots,b_{\mu})$ be a basis of $R / \Qc_{\bm \zeta}$ and
$\DD=(\Lambda_1,\ldots, \Lambda_{\mu})$ its dual counterpart, with $\Lambda_1=\bm 1$.

We introduce a new set of equations starting from $\bm f$, as follows: add
for every $f_i$ the polynomial $ g_k = f_k + p_k$, $p_k= \sum_{i=1}^\mu
\ee_{i,k}b_{i} $ where $\bm \ee_k= (\ee_{1,k},\dots,\ee_{\mu,k})$ is a new vector of $\mu$ variables.

Consider the system 
$$
\DD \bm g(\bm x,\bm \ee)= \Big( \Lambda_1(\pb_{\xb})[\gb], \dots, \Lambda_{\mu}(\pb_{\xb})[\gb]  \Big).
$$
where $\Lambda^{\xb}[g_k] =\Lambda_i(\pb_{\xb})[g_k]$ is defined as in \eqref{eq:defdual} with
$\bm \zeta$ replaced by $\xb$, ie. we differentiate $g_k$ but we do not evaluate at $\bm \zeta$.
This is a system of $\mu s$ equations, which we shall index $\DD\bm
g(\bm x,\bm \ee) =( g_{1,1}, \dots, g_{\mu,s} )$. We have
$$
g_{ik}(\xb,\bm \ee)= \Lambda_{i}^{\xb} [ f_{k} + p_k ] =
\Lambda_{i}^{\xb} [ f_{k} ]  + \Lambda_{i}^{\xb}[ p_k ]=
\Lambda_{i}^{\xb} [ f_{k} ]  + p_{ik}(\xb,\bm \ee).
$$
Notice that $p_{i,k}(\bm \zeta,\bm \ee)=\Lambda_{i}^{\bm \zeta}[p_k]=\ee_{i,k}$ because
$\DD=(\Lambda_{1},..,\Lambda_\mu)$ is dual to $\Bc$. 

As the first basis element of $\DD$ is $\bm 1$ (the evaluation
at the root), the first $s$ equations are $\bm g(\xb,\bm \varepsilon)=0$.

Note that this system is under-determined, since the number of
variables is $\mu \, s + n$ and the number of equations is $\mu s$.
We shall provide a systematic way to choose $n$ variables and purge
them (or better, set them equal to zero).

This way we arrive to a square system $\DD\bm{g}(\bm x,
\bm{\tilde\ee})$ (we use $\bm{\tilde \ee}$ for the remaining $\mu s 
- n$ variables) of size $\mu s \times \mu s$. We shall prove that this
system vanishes on $({\bm\zeta},\bm 0)$ and that $J_{\DD\bm{g}}(\bm \zeta,\bm
0) \neq 0$. 

By linearity of the Jacobian matrix we have
\begin{align*}
J_{\DD\bm{g}}(\bm x, \bm \ee) &=
J_{\DD\bm f}(\bm x,\bm \ee) +  J_{\DD\bm p}(\xb,\bm\ee)\\
&= [\, {J}_{\DD\bm f}(\bm x) \,|\,\bm 0\ \,] + 
[\, {J}_{\DD\bm p}^{\bm x}(\bm x,\bm\ee) \, |\, {J}_{\DD\bm p}^{\bm \ee}(\bm x,\bm\ee)\, ]  ,
\end{align*}
where ${J}^{\xb}_{\DD\bm p}(\bm x,\bm \ee)$ (resp. $J^{\bm
  \ee}_{\DD\bm p}(\xb,\bm\ee)$) is the Jacobian matrix of $\DD\bm p$
with respect to $\xb$ (resp. $\bm \ee$).
\begin{lemma}\label{jacp}
  The Jacobian $J^{\bm \ee}_{\DD\bm p}(\bm x,\bm\ee)$ of the linear system $\DD\bm
  p=(p_{1,1},\dots,p_{\mu , s})$ with
  $p_{i,k}(\bm\ee_k)=\Lambda_{i}^{\xb} [p_k](\bm x, \bm \ee_k)$ 
  evaluated at $(\bm x, \bm\ee)=(\zeta,\bm 0)$ is the
  identity matrix 
  in dimension $\mu s$.
\end{lemma}
\begin{proof}[of Lemma~\ref{jacp}]
  First note that the system is block se\-pa\-ra\-ted, i.e. every
  $p_{ik}$ depends only on variables $\bm\ee_i$ and not on all
  variables $\bm \ee= (\bm\ee_1,\dots,\bm\ee_n)$. This shows that
  $J_{\bm p}^{\bm \ee}(\bm x,\bm\ee)$ is block diagonal,
{\small
  $$J^{\bm \ee}_{\DD\bm p}(\xb,\bm\ee) = 
  \left[
    \begin{array}{ccc}
      J_{1} &        & 0  \\
        & \ddots &    \\
      0 &        &  J_{\mu} 
    \end{array}
  \right].
  $$
}
  Now we claim that all these blocks are all equal to the identity
  matrix. To see this, consider their entry
  $ \partial_{\ee_{kj}}[p_{ik}]$ for $i,j=1, \dots, \mu $, which is
  $$
  \pp{\ee_{kj}}\Lambda_{i-1}^{\xb} [p_{k}]=  \Lambda_{i}^{\xb}[\pp{\ee_{jk}}p_{k} ]
  =\Lambda_{i}[ b_{j}] =
  \left\{\begin{array}{ll}1& ,i=j\\0& ,\text{otherwise} \end{array}\right., 
  $$
  since $\pp{\ee_{j,k}}p_{k}= \pp{\ee_{j,k}}(b_{j}\ee_{j,k}) = b_{j}$.
\end{proof}

\begin{lemma}\label{jacx}
  The $\mu s \times n$ Jacobian matrix $J_{\DD \bm f}(\bm x)$ of
  the system $\DD\bm f(\bm x)= ( f_1,\dots, f_{\mu n})$ % in variables $\bm x$
  is of full rank $n$ at $\xb=\bm\zeta$.
\end{lemma}
\begin{proof}[of Lemma~\ref{jacx}]
  Suppose that the matrix is rank-deficient. Then there is a non-trivial
  vector in its kernel,
  $$
  J_{\DD\bm{f}}(\bm\zeta) \cdot \bm v = \bm 0  .
  $$
  The entries of $\bm v$ are indexed by $\partial_{i}$.  This implies that a
  non-zero differential $\Delta= v_1 \partial_1 + \cdots + v_n \partial_n $ of order
  one satisfies the following relations:
  $
  (\Delta\Lambda_{i})^{\bm\zeta}[f_j]= 0,  i=1,\ldots,\mu, j=1,\ldots,s.
  $
  By the standard derivation rules, we have
  $$
  \pp{\partial_k} (\Delta\Lambda_{i})= 
  v_k\Lambda_{i} + \Delta \pp{\partial_k}\Lambda_{i}  ,
  $$
  for $i=1,\ldots,\mu, ,k=1,\dots,n.$
  Since $\DDD$ is stable by derivation, $\pp{\partial_k}\Lambda_{i}\in
  \DDD$.
  We deduce that the vector space spanned by $\<\DDD,\Delta\DDD\>$ is stable
  by derivation and vanishes on $\bm f$ at $\bm \zeta$.
  By Proposition~\ref{prop:diffpol}, we deduce that $\Delta\DDD\subset
  \DDD$. This is a contradiction, since $\Delta$ is of degree $1$ and the
the elements in $\DDD$ are of degree $\le N$.
\end{proof}
The columns of $J_{\DD\bm g}(\bm x, \bm \ee)$ are indexed by the
variables $(\bm x, \bm \ee)$, while the rows are indexed by the
polynomials $g_{ik}$.  We construct the following systems:
\begin{itemize}
 \item[(a)] Let $\DD\fb^{I}$ be a subsystem of $\DD\fb$ such that the corresponding
$n$ rows of $J_{\DD\fb}(\bm\zeta)$ are linearly independent
   (Lem. \ref{jacx}). 
   We denote by $I=\{(i_1,k_{1}),\ldots,(i_n,k_{n})\}$ their indices.

 \item[(b)] Let ${\DD\tilde \gb}(\xb,\tilde{\bm\ee})$ be the square
   system formed by re\-mo\-ving the variables
   $\ee_{i_{1},k_{1}},\ldots,\ee_{i_{n},k_{n}}$ from ${\DD \gb}(\xb,{\bm\ee})$. %$\bm \ee$.
   Therefore the Jacobian $J_{\DD\tilde \gb}(\xb,\tilde{\bm\ee})$ derives from 
   $J_{\DD \bm g}(\bm x,\bm\ee)$, after purging the columns indexed by
   $\ee_{i_{1},k_{1}},\ldots,\ee_{i_{n},k_{n}}$,
   and it's $(i_j,k_{j})$ row becomes $[\nabla
   (\Lambda_{i_{j}}^{\xb} \tilde g_{i_{j},k_{j}})^T |\ \bm 0\ ]$.
\end{itemize}

\begin{theorem}[Deflation Theorem 1] \label{deflation1}
  Let $\fb(\bm x)$ be a $n-$variate polynomial system with an $\mu -$fold
  isolated zero at $\xb=\bm \zeta$.
  Then the $n\times n$ system ${\DD\fb^{I}}(\xb)= 0$, defined in (a), has a simple root
  at $\xb={\bm\zeta}$.
\end{theorem}
\begin{proof}
By construction, ${\bm\zeta}$ is a solution of
${\DD\fb^{I}}(\xb)=0$. Moreover, the indices $I$ are chosen such that
$
\det J_{\DD\fb^{I}}(\bm\zeta) \neq 0.
$
This shows that $\bm\zeta$ is a simple (thus isolated) root
  of the system ${\DD\fb^{I}}(\xb)= 0$.
\end{proof}
\begin{example}
  In our running example, we expand the rectangular Jacobian matrix of $6$
  polynomials in $(x_1,x_2)$. Choosing the rows corresponding to
  $f_1$ and $(d_1-d_2^2-d_1d_2-d_1^2)[f_1]$, we find a non-singular minor,
  hence the resulting system $(f_1, 2x_1)$ has a regular root at
  $\bm\zeta=(0,0)$.
\end{example}

The deflated system ${\DD\fb^{I}}(\xb)=0$ is a square system in $n$
variables. Contrarily to the deflation approach
in~\cite{lvz06,zeng05}, we do not need to introduce new variables here
and one step of deflation is sufficient.
In the following theorem, we do introduce new variables to express the
condition that the perturbed system has a given multiplicity
structure.
\begin{theorem}[Deflation Theorem 2] 
  \label{deflation2} Let $\fb(\bm x)$ be a $n-$variate polynomial
  system with an $\mu -$fold isolated zero at $\bm x=\bm \zeta$.
  The \emph{square} system ${\DD\tilde\gb}(\xb,\tilde{\bm\ee})=0$, as defined in (b), has
  a regular isolated root at $(\bm x,\bm{\tilde \ee})=(\bm\zeta,\bm
  0)$.
\end{theorem}
\begin{proof}
By definition of $\DD$, we have 
$$ 
{\DD\tilde\gb}(\bm\zeta,\bm 0)= (\Lambda^{\bm\zeta}_{1}[\fb], \ldots,\Lambda_{\mu}^{\bm\zeta}[\fb]) =0.
$$
Moreover, by construction of $\DD\tilde\gb$ we get, up to a row permutation, the determinant:
  \[
  \pm \text{det } J_{\DD\tilde\gb}(\bm \zeta,\bm 0)= \text{det}
  \left|\begin{array}{cc}
      J_1 & 0 \\
      J_2 & I 
    \end{array}\right| =
  \text{det } J_1 \ne 0 ,
  \] 
where $J_{1}=J_{\DD\fb^{I}}(\bm\zeta)$.
This shows that $(\bm\zeta,\bm 0)$ is regular and thus isolated
point of the algebraic variety defined by ${\DD\tilde{\gb}}(\xb,\tilde{\bm\ee})=0$. 
\end{proof}

Nevertheless, this deflation does differ from the deflation strategy
in~\cite{lvz06,zeng05}. There, new variables are added that
correspond to coefficients of differential elements,
 thus in\-tro\-du\-cing a perturbation in the approximate dual basis,
in case of approximate input. Hence the output concerns a deflated root of
the given approximate system.
In our  method, we perturb the equations, keeping an
approximate structure of a multiple point. Consequently, the
certification of a root concerns a nearby system, within controlled
error bounds, as it shall be described in Sect.~\ref{sec:verify}.

We mention that it would also be possible to use the equations
~(\ref{eq:candidate}, i-iii) to construct a deflated system on the
differentials and to perturb the approximate dual structure.

\section{\hspace{-.2cm}Verifying approximate singular points}\label{sec:verify}

In real-life applications it is common to work with approximate
inputs.  Also, there is the need to (numerically) decide if an
(approximate) system possesses a single (real) root in a given domain,
notably for use in subdivision-based algorithms,
eg.~\cite{mp09,mmt09}.

In the regular case, Smale's $\alpha-$theory, extending Newton's
method, can be used to answer this problem.  Another option is Rump's
Theorem, also based on Newton theory. In our implementation we choose
this latter approach, since it is suitable for inexact data and suits
best with the perturbation which is applied. Our perturbation coincides
to the numerical scheme of~\cite{RG10} in the univariate
case.

The certification test is based on the verification method of
Rump~\cite[Th.~2.1]{RG10}, which we rewrite in our setting:
\begin{theorem}[\cite{RG10} Rump's Theorem]\label{th:rump}
  Let $\fb\in R^n $ be a polynomial system and ${\bm\zeta}^*\in\mathbb R^n$ a real
  point.  Given an interval domain $Z\in \mathbb{IR}^n$ containing
  ${\bm\zeta}^*\in\mathbb R^n$, and an interval matrix
  $M\in\mathbb{IR}^{n\times n}$ whose $i-$th column $M_i$ satisfies $
  \nabla f_i( Z ) \subseteq M_i $ for $i=1\dots,n$, then the
  following holds:

If the inclusion
\begin{align}\label{inclusion}
V(\bm f,Z, \bm\zeta^* )=
- J_{\fb}({\bm\zeta}^*)^{-1} \fb({\bm\zeta}^*) +  
(I - J_{\fb}({\bm\zeta}^*)^{-1} M ) Z 
\subseteq   \stackrel{\circ}{ Z }
\end{align}
is true, then there is a unique $\bm\zeta \in Z$ with $\fb(\bm\zeta)=
\bm 0$ and the Jacobian matrix $J_{\fb}(\bm \zeta)\in M $ is non-singular.
\end{theorem}

  This theorem is applied on the perturbed system. If the test
  succeeds, we also get a domain for $\bm \varepsilon-$variables that
  reflects the distance of the approximate system from a precise
  system with the computed local structure.

\begin{example}
  We start with an approximate system: $f_1= 1.071 x_1-1.069 x_2+1.018
  x_1^2,\, f_2= 1.024 x_1-1.016 x_2+1.058 x_2^2$ and the domain:
  $Z=[-.01, .03]\times [-.03, .01]$. The Jacobian at $\bm x=(0,0)$
  evaluates to $.00652$, hence it is close to singular.

  We set the tolerance equal to $.04$, i.e. the size of the domain, and
  we consider the center of the box as our ap\-pro\-xi\-ma\-te point,
  $\bm\zeta^* =(.01,-.01)$.

First we compute approximate multiplicity structure at $\bm\zeta^*$,
$\DD=(
1, \,
d_2+1.00016 d_2^2+.99542 d_1d_2+1.03023 d_1^2, \,
d_1-1.00492 d_2^2-1.00016 d_1d_2-1.03514 d_1^2
)$
as well as $(1,x_2,x_1)$, a standard basis for the quotient. 
The latter indicates to perturb up to linear terms.

Now apply the second deflation theorem~\ref{deflation2} to get the $6\times 6$ system
$
\bm g= (
1.018\,x_{{1}}^{2}+ 1.071\,x_{{1}}+ \left( \varepsilon_{12} - 1.069\right) x_{{2}}, \,
\varepsilon_{12} - .02023,\,
.01723+ 2.036\,x_{{1}} ,\,
 1.058\,x_{{2}}^{2}
+(  1.024+ \varepsilon_{23} ) x_{{1}}
+ (  \varepsilon_{22} - 1.016 ) x_{{2}}
+ \varepsilon_{21},\,
.04217+ 2.116\,x_{{2}}+ \varepsilon_{22}, \,
\varepsilon_{23} - .03921 )
$,
which has a regular root for $\bm \zeta\in Z$ and parameters $(\varepsilon_{12},\varepsilon_{21},\varepsilon_{22},\varepsilon_{23})$. 
Indeed, applying
Theorem~\ref{th:rump} with $Z'=Z \times
[-.04,.04]^4$ and $(\bm\zeta^*,0,..,0)$ we get an inclusion %verification
$V(\bm g,Z' , \bm\zeta^* )\subseteq \stackrel{\circ}{Z'} $.
%inclusions~\eqref{inclusion} in the interior of $Z$.
% (-0.846e-2 , -0.846e-2)
% (-0.1118e-1 , -0.667e-2)
% 0.02023
% (-0.745e-2 , 0.523e-2)
% (-0.2806e-1 , -0.1850e-1)
% 0.03921
\end{example}

\section{\hspace{-.1cm}Geometry around a singularity}\label{applications}

As a final step in analyzing isolated singularities, we show how the
local basis can be used to compute the topological degree around the
singular point. If the latter is a self-intersection point of a
real algebraic curve, one can deduce the number of curve branches that
pass through it.

\smallskip\noindent{}\textbf{Topological degree computation.}
Let $\fb(\bm x)$ be a square $n-$variate system with an $\mu -$fold
isolated zero at $\xb=\bm \zeta$.
To a linear form $\Lambda\in \mathbb R[\pb_{\bm \zeta}]$, we associate the
%real
quadratic form
 \begin{align} 
   Q_\Lambda \ : \ R/\Qc \times R/\Qc \to \mathbb R  \label{eq:qform} \ \ \ , \ \ \ 
   (b_i, b_j ) \mapsto  \Lambda (b_i b_j) 
 \end{align}
for $R/\Qc=\id{b_1,\dots,b_\mu}$.
 The signature of this (symmetric and bi-linear) form is
the sum of signs of the diagonal entries of any diagonal matrix
representation of it.
 
\begin{proposition}[{\cite[Th.~1.2]{el77}}] \label{th:tdeg}
  If $Q_\Phi$, $\Phi\in \DD$ is any bi-linear symmetric form such
  that $\Phi^{\bm \zeta} [\det J_{\fb}(\xb)] >0$, then
  \begin{align}
  \text{tdeg}_{\bm \zeta}(\fb) = sgn(Q_{\Phi}).
  \end{align}
\end{proposition}
This signature is independent of the bi-linear form used.

We can use this result to compute the topological degree at $\xb =\bm \zeta$
using the dual structure at $\bm \zeta$. Since a basis $\DD$ is available we set
$\Phi=\pm \Lambda_i$, for some basis element that is not zero on
$\det J_{\fb}(\xb)$. Indeed, such an element can be retrieved among the basis
elements, since $\det J_{\fb}\notin \id{\bm f}$, see \cite[Ch.~0]{irsea-em-07}.

In practice it suffices to generate a random element of $\DDD$,
compute it's matrix representation 
$[\Phi(b_ib_j)]_{ij}$,
 and then extract the signature of $Q_{\Phi}$.

\smallskip\noindent{}\textbf{Branches around a singularity.}
In the context of computing with real algebraic curves, the
identification of self-intersection points is only the first step of
determining the local topology.  As a second step, one needs to
calculate the number of branches attached to the singular point $\bm \zeta$,
hereafter denoted $\text{Br}(\bm f,\bm \zeta)$.  This information is
encoded in the topological degree.

An implicit curve in $n-$space is given by all points sa\-ti\-sfy\-ing $\bm f(\xb)=0$,
$\bm f=(f_1,\dots,f_{n-1})$.  Consider $p(\bm x)= (x_1-\zeta_1)^2
+\cdots + (x_n-\zeta_n)^2$, and $g(\bm x)= \det J_{(\bm f,p)}(\xb)$. 
Then~(\cite{sza99} and references therein):
\begin{align}
\text{Br}(\bm f,\bm \zeta) = 2\, \text{tdeg}_{\bm \zeta} (\bm f, g). 
\end{align}
This implies an algorithm for $\text{Br}(\bm f,\bm \zeta)$. 
First compute the primal-dual structure of $(\bm f, g)$ at $\bm
\zeta$ and then use Prop.~\ref{th:tdeg} to get $\text{tdeg}_{\bm \zeta} (\bm f, g)$.

\begin{example}\label{ex:branches}
  Consider the implicit curve $f(x,y)=0$, in $xy-$plane, with
  $f(x,y) = x^4+2 x^2 y^2+y^4+3 x^2 y-y^3$, that looks like
  this \includegraphics[scale=0.17]{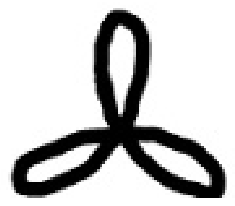}. We search
  for the number of branches touching $\bm \zeta= (0,0)$.

  We compute $g(x,y)=J_{(f,x^2+y^2)}= 18xy^2-6x^3$, and then the
  multiplicity structure of $(f,g)$ at $\bm \zeta$, and we arrive to
  the standard basis $\Bc=(1, y, x, y^2, xy, x^2, y^3, xy^2, x^2y)$.
  Among the $9$ elements of the dual basis, we find $\Phi= d_y^3+\frac
  3 8 d_y^4+ \frac 1 8 d_x^2 d_y^2+ \frac 3 8 d_x^4 $, having value
  $\Phi^{\bm 0}[ \det J_{(f,g)}(\bm x)]=54 > 0 $ on the Jacobian
  determinant.

  Using $\Bc$ and~\eqref{eq:normalf}, we get the $9\times 9$ matrix
  representation of $Q_{\Phi}$~\eqref{eq:qform} with $ij-$th entry
  $\Phi[ \text{NF}(\xb^{\bm \beta_i}\xb^{\bm \beta_j}) ]$, and we
  compute $\text{tdeg}_{\bm \zeta}\,(f,g)= \text{sgn}\, Q_{\Phi} = 3$,
  thus $\text{Br}(f,(0,0))=6$.
\end{example}
\section{Experimentation}\label{examples}

Our method is developed in \Maple. It
uses Mourrain's Integration technique to compute (approximate) dual
basis and derive the augmented system of Th.~\ref{deflation2}. 
Then Rump's method is used to verify the root.
Macaulay's method is also implemented for testing purposes.

\begin{example}
  Consider the system~\cite{lvz08} of $3$ equations in $2$ variables $f_1 =x_1^3
  +x_1 x_2^2 ,\, f_2 =x_1 x_2^2 + x_2^3 ,\, f_3 =x_1^2x_2 + x_1 x_2^2$,
  and the singular point $(0,0)$.% of multiplicity $7$.

  Suppose that the point is given. Using~\ref{th:integrate}
  and~\ref{lem:tbasis} we derive the primal-dual pair $\DD= ({1},
  {d_1}, {d_2},{d_1^2}, {d_1d_2}, {d_2^2},
  \underline{d_2^3}+d_1^3+d_1^2d_2-d_1d_2^2)$, where $d_2^3$ is
  underlined to show that it corresponds to $x_2^3$ in the primal
  standard basis
  $\Bc$.  %$\Bc= (1, x_1, x_2, x_1^2, x_1x_2, x_2^2, x_2^3) $
  The biggest matrix used, in depth $4$, was of size $9 \times 8 $,
  while Macaulay's method terminates with a matrix of size $30\times
  15$.

  To deflate the root, we construct the augmented system $\DD\bm f$ of
  $21$ equations. The $21\times 2$ Jacobian matrix $J_{\DD\bm f}^{\bm
    x}$ is of rank $2$ and a full-rank minor consists of the rows
  $4$ and $5$. Consequently find the system $( d_1^2[f_1], d_1d_2[f_1]
  )=(3x_1,2x_2)$ which deflates $(0,0)$. Note that even though both
  equations of the deflated system derive from $f_1$, the functionals
  used on $f_1$ are computed using all initial equations.
\end{example}

\begin{example}
Let, as in~\cite{Lecerf02,lvz08}, 
$f_1 = 2 x_1 + 2 x_1^2 + 2 x_2 + 2x_2^2 + x_3^2 -1,\,
f_2 = ( x_1 + x_2 - x_3 -1 )^3 - x_1^3$, and 
$f_3 =   2 x_1^3 + 2x_2^2 + 10x_3 + 5x_3^2 + 5)^3 - 1000 x_1^5 $.

Point $(0,0,-1)$ occurs with multiplicity equal to $18$, in depth $7$.
The final matrix size with our method is $206\times 45$, while
Macaulay's method ends with a $360\times 165$ matrix.

If the objective is to deflate as efficiently as possible, then
one can go step by step: First compute a basis for $\DDD_1$ and stop
the process. We get the evaluation $\bm 1$ and $2$ first order
functionals, which we apply to $f_1$. We arrive to $(\bm
1[f_1],(d_2-d_1)[f_1]$, $(d_1+d_3)[f_1]) = (f_1, -4x_1+4x_2,
2+4x_1+2x_3 )$ and we check that the Jacobian determinant is $64$,
thus we have a deflated system 
only with a partial local
structure.
\end{example}
Table~\ref{tab:bench} shows computations on
the benchmark set of~\cite{zeng05}.  Multiplicity, matrix sizes at
termination step are reported.
\begin{table}[ht]\label{tab:bench} 
 \centering
\begin{tabular}{|l|c|c|c|}\hline
\bf Sys.&$\bm \mu$&\bf Integration  &\bf Macaulay \\ \hline\hline
cmbs1   & 11 &$33\times 23$ &$105 \times 56 $   \\ \hline
cmbs2   & 8  &$21 \times 17$&$60 \times 35 $   \\ \hline
mth191  & 4  &$10 \times 9$ &$30 \times 20 $   \\ \hline
decker2 & 4  &$5 \times  5$ &$20 \times 15 $   \\ \hline
Ojika2  & 2  &$6 \times  5$ &$12 \times 10 $   \\ \hline
Ojika3  & 4  &$24 \times 9 $&$60 \times 35 $   \\ \hline
KSS     & 16 &$569\times 69$&$630 \times 252 $   \\ \hline
\small 
Caprasse& 4  &$34 \times 13$&$60 \times 35 $   \\ \hline
Cyclic 9& 4  &$ 369 \times 33 $ &$495 \times 33 $   \\ \hline
DZ1     &131 &$1450 \times 524$ &$ 4004\times 1365$   \\ \hline
DZ2     & 7  &$73 \times 33 $ &$ 360\times 165 $    \\ \hline
DZ3     & 5  &$ 14\times 8 $ &$ 30 \times 21 $   \\  \hline
\end{tabular}
\caption{Benchmark systems from [3].}
\vspace{-.3cm}
\end{table}

 \smallskip\noindent{}\textbf{Acknowledgments.}
 This research has received funding from the
 EU's 7\textsuperscript{th} Framework Programme [FP7/2007-2013],
 Marie Curie ITN SAGA, grant n\textsuperscript{o}
 [PITN-GA-2008-214584].

\appendix
We attach additional examples that did not fit page limits.

\vspace{.2cm}

\textsc{Verification Example.}

  Let $f_1=(x_1^2x_2-x_1x_2^2, f_2= x_1-x_2^2)$. The verification
  method of~\cite{RG10} applies a linear perturbation on this system,
  but fails to certify the root $\bm x=(0,0)$.

  We consider an approximate point $\bm\zeta^* = (.01,\, .002)$ and we
  compute the approximate multiplicity structure:
  $$
  \DD= (\Lambda_1,\dots,\Lambda_4)=( 1.0, 1.0 d_2, \underline{1.0 d_1}+1.0 d_2^2,  \underline{1.0 d_1 d_2} +1.0 d_2^3 )
  $$
  The augmented system $\bm g(\bm x)=(\Lambda_j(f_i))=(f_1,\, 2.0 x_1x_2-1.0x_2^2-1.0x_1,\, 2.0x_1-2.0x_2,\, 1.0x_1-1.0x_2^2,\, f_2,\, -2.0x_2,\, 0., 0.)$
  has a Jacobian matrix:
 $$ J_{g}(\bm\zeta^*)^{T}= 
 \left[ \begin {array}{cccccccc}  
 .00& .016&- .99 & 2.0& 1.0&0&0&0\\ 
 .00&- .02& .016&- 2.0&-.004&- 2.0&0&0
 \end{array} 
 \right] 
 $$
 with a non-zero minor at the third and forth row. Using this information, we apply the following perturbation to the original system:
  \begin{align*}
  f_1 &= x_1^2x_2-x_1x_2^2 + \varepsilon_{11} + \varepsilon_{12} x_2 \\
  f_5 &= x_1-x_2^2  + \varepsilon_{21} + \varepsilon_{22} x_2 + \varepsilon_{23} x_1 + \varepsilon_{24} x_1x_2
  \end{align*}
  Thus $\bm g(x_1, x_2, \varepsilon_{11},\varepsilon_{12},
  \varepsilon_{21}, \varepsilon_{22}, \varepsilon_{23},
  \varepsilon_{24} )$, computed as before, is a square system with
  additional equations
  \begin{align*}
%f_1 &= 1.0x_1^2x_2-1.0x_1x_2^2+1.0\varepsilon_{11}+1.0x_2\varepsilon_{12} \\
f_2 &= 1.0x_1^2-2.0x_1x_2+1.0\varepsilon_{12} \\
f_3 &= 2.0x_1x_2-1.0x_2^2-1.0x_1  \\
f_4 &= 2.0x_1-2.0x_2 \\
%f_5 &= 1.0x_1-1.0x_2^2+1.0\varepsilon_{21}+1.0x_2\varepsilon_{22}+1.0x_1\varepsilon_{23}+1.0x_1x_2\varepsilon_{24} \\
f_6 &= -2.0x_2+1.0\varepsilon_{22}+1.0x_1\varepsilon_{24} \\
f_7 &= 1.0\varepsilon_{23}+1.0x_2\varepsilon_{24} \\
f_8 &= 1.0 \varepsilon_{24}
\end{align*}

Now take the box $Z_1=[-.03, .05]\times [-.04, .04]\times [-.01,.01]^6$.
We apply Th.~\ref{th:rump} on $\bm g$, ie. we compute 
$V(\bm g,Z_1 , \bm\zeta^* )$.
For the variable $\varepsilon_{21}$ the interval is $[-.015,.15] \not\subseteq (-.01,.01) $, therefore we don't get an answer.

We shrink a little $Z_1$ down to $Z_2=[-.03, .05]\times [-.02, .02]\times [-.01,.01]^6$ and we apply again Th.~\ref{th:rump}, which results in
$$
V(\bm g,Z_2 , \bm\zeta^* ) =
\left[
\begin{array}{c}
{[-.004,.004]} \\
{[-.004,.004]}\\
{[-.001,.001]}\\
{[-.007,.007]}\\
{[-.006,.006]}\\
{[-.009,.009]}\\
{[-.00045, .00035]} \\
{[.0,.0]} 
\end{array}
\right] \subseteq  \stackrel{\circ} Z_2  ,
$$
thus we certify the multiple root of the original system inside $Z_2$.
\qed

\[ \small
\left[ 
\begin {array}{ccccccccc} 
1&y&x&{y}^{2}&xy&{x}^{2}&{y}^{3}&x{y}^{2}&{x}^{2}y\\
y&{y}^{2}&xy&{y}^{3}&x{y}^{2}&{x}^{2}y&\frac 3 8 {y}^{3}- \frac{9}{8} {x}^{2}y&0&\frac 1 8 {y}^{3}-\frac 3 8 x^{2}y \\ 
x&xy&{x}^{2}&x{y}^{2}&{x}^{2}y&3x{y}^{2}&0&1/8{y}^{3}-3/8{x}^{2}y&0\\ 
{y}^{2}&{y}^{3}&x{y}^{2} &\frac 3 8{y}^{3}-{\frac {9}{8}}{x}^{2}y&0&\frac 1 8{y}^{3}-\frac 3 8{x}^{2}y&0&0 &0\\ 
xy&x{y}^{2}&{x}^{2}y&0&\frac 1 8{y}^{3}-\frac 3 8{x}^{2}y&0&0&0&0 \\
{x}^{2}&{x}^{2}y&3x{y}^{2}& \frac 1 8{y}^{3}-3/8{x}^{2}y&0&3/8{y}^{3}-{\frac {9}{8}}{x}^{2}y&0&0&0\\ 
{y}^{3}&3/8{y}^{3}-{\frac{9}{8}}{x}^{2}y&0&0&0&0&0&0&0\\
x{y}^{2}&0&\frac 1 8{y}^{3}-\frac 3 8{x}^{2}y&0&0&0&0&0&0\\
{x}^{2}y&1/8{y}^{3}-\frac 3 8{x}^{2}y&0&0&0&0&0&0&0
\end{array} 
\right] 
\]\\
{\bf Multiplication table for $R/\id{f,g}$ (Example~\ref{ex:branches}). }

\pagebreak

\textsc{Example~\ref{ex:branches} (Cont'd).}

  Consider the implicit curve $f(x,y)=0$, in $xy-$plane, with
  $$
  f(x,y) = x^4+2 x^2 y^2+y^4+3 x^2 y-y^3, 
  $$ 
  that looks like
  this \includegraphics[scale=0.17]{figs/sing_curve.eps}. We search
  for the number of branches touching $\bm \zeta= (0,0)$.

  This point is of multiplicity $4$, as the dual basis we get for 
  $$
  f(x,y)=\pp{x}{f}(x,y)=\pp{y}{f}(x,y)=0
  $$ 
  is $(1, d_x, d_y, d_x^2+d_y^2)$, which provides no information for
  the number of branches.

  We compute 
  $$
  g(x,y)=J_{(f,x^2+y^2)}= 18xy^2-6x^3 , 
  $$ 
  and then the multiplicity structure of $(f,g)$ at $\bm \zeta$, and
  we arrive to
\begin{align*}
\DD= & ( 1,d_{{y}},d_{{x}},d_{{y}}^{2},d_{{x}}d_{{y}},d_{{x}}^{2},\underline{d_{{y}}
^{3}}+\frac 3 8\,d_{{y}}^{4}+\frac 1 8\,d_{{x}}^{2}d_{{y}}^{2}+\frac 3 8\,
d_{{x}}^{4}, \\
& \underline{d_{{x}}d_{{y}}^2}+3\,d_{{x}}^{3},\underline{ d_{{x}}^{2}d_{{y}} }-
\frac {9}{8}\,d_{{y}}^{4}-\frac 3 8\, d_{{x}}^{2}d_{{y}}^{2}-\frac {9
}{8}\,d_{{x}}^{4} ) ,
\end{align*}
  and the standard basis 
  $$
  \Bc=(1, y, x, y^2, xy, x^2, y^3, xy^2, x^2y) .
  $$
  Among the $9$ elements of the dual basis, we find 
  $$ 
  \Phi= d_y^3+\frac
  3 8 d_y^4+ \frac 1 8 d_x^2 d_y^2+ \frac 3 8 d_x^4 ,
  $$ 
  having value $\Phi^{\bm 0}[ \det J_{(f,g)}(\bm x)]=54 > 0 $ on the Jacobian
  determinant.

  Using $\Bc$ and~\eqref{eq:normalf}, we compute the matrix $\text{NF}(\xb^{\bm \beta_i}\xb^{\bm \beta_j})$ of multiplication  in $R/\id{f,g}$, given at the end of the page.
 Now a representation of $Q_{\Phi}$~\eqref{eq:qform} can be computed, 
 by applying $\Phi^{\bm 0}$ on the multiplication table to get:
\[ \small Q_{\Phi}=
 \left[ \begin {array}{ccccccccc} 0&0&0&0&0&0&1&0&0\\ 
0&0&0&1&0&0&3/8&0&1/8\\ 
0&0&0&0
&0&0&0&1/8&0\\ 
0&1&0&3/8&0&1/8&0&0&0\\ 
0&0&0&0&1/8&0&0&0&0\\ 
0&0&0&1/8
&0&3/8&0&0&0\\ 
1&3/8&0&0&0&0&0&0&0
\\ 0&0&1/8&0&0&0&0&0&0\\ 
0&1/8&0&0
&0&0&0&0&0\end {array} \right].
\]

With a QR iteration, we find 6 positive and 3 negative
eigenvalues of this representation, hence we compute 
$$\text{tdeg}_{\bm \zeta}\,(f,g)=
\text{sgn}\, Q_{\Phi} = 6 - 3 = 3 ,
$$ i.e. there are $6$ branches of the curve
around $(0,0)$.  \qed

\end{document}